\documentclass[aps,twocolumn,pra,groupedaddress,superscriptaddress,nofootinbib]{revtex4-2}

\usepackage{amsthm}
\usepackage{graphicx,amsmath,amssymb,amsbsy,subfigure,hyperref,bbm,times,txfonts,float}
\usepackage{subfigure,hyperref,bbm,times}
\usepackage[T1]{fontenc}
\usepackage{braket}
\usepackage{amsthm}

\usepackage{epsfig} 
\usepackage{color}
\usepackage{graphicx}% Include figure files
\usepackage{dcolumn}
\usepackage{bm}% bold math
\usepackage[normalem]{ulem}

\newtheorem{definition}{Definition}
\newtheorem{lemma}{Lemma}

\newtheorem{conjecture}{Conjecture}
\newtheorem{observation}{Observation}

\providecommand{\openone}{\leavevmode\hbox{\small1\kern-3.8pt\normalsize1}}

\usepackage{soul}

\hypersetup{
   colorlinks=true,
   linkcolor=blue,
}
\graphicspath{{figure/}}

\begin{document}

\title{Variational Approach for Uniform Quantum Permutation Generators}

\author{Farzam Nosrati}
\affiliation{Nexus for Quantum Technologies, University of Ottawa, K1N 5N6, Ottawa, ON, Canada}
\affiliation{IMDEA Networks Institute, Madrid, Spain}
%\orcid{}

\author{Nicolás Borrajo}
\affiliation{IMDEA Networks Institute, Madrid, Spain}
\affiliation{ETH, Zurich, Switzerland}

\author{Antonio {Fernández Anta}}
\affiliation{IMDEA Networks Institute, Madrid, Spain}
\affiliation{IMDEA Software Institute, Madrid, Spain}

\author{Vincenzo Mancuso}
\affiliation{Dipartimento di Ingegneria, Universit\`{a} degli Studi di Palermo, Viale delle Scienze, 90128 Palermo, Italy}
\affiliation{IMDEA Networks Institute, Madrid, Spain}
\begin{abstract}
Uniform permutation generation is a fundamental task in both classical and quantum computation, with applications ranging from cryptography to quantum optimization and quantum error correction. Existing exact quantum constructions typically require all-to-all qubit connectivity and quadratic circuit depth. We develop a variational quantum circuit framework for uniform permutation generation under connectivity constraints, in which the circuit architecture is determined by the underlying interaction graph and the variational parameters are optimized to enforce the target permutation statistics. In particular, we present explicit controlled-SWAP-based unitary constructions that achieve exact uniformity with quadratic circuit size and linear depth \(O(n)\) on linear nearest-neighbor topologies. Our approach, therefore, removes the need for all-to-all connectivity while improving the depth of previous exact constructions by a factor. We further prove that a quantum Bene\v{s}-like architecture is intrinsically non-uniform. Despite its logarithmic depth and ability to realize any permutation it cannot generate a uniform distribution over permutations for any choice of variational parameters. These results clarify the role of circuit topology in exact permutation generation and identify variational quantum circuits as a natural framework for hardware-constrained uniform sampling. More broadly, this work suggests that exact uniform permutation generation is a strictly stronger requirement than mere permutation realizability, and lays the groundwork for a formal complexity separation between the two.
\end{abstract}

\maketitle

\section{Introduction}

Random permutation generation is a fundamental primitive in classical computer science. Its objective is creating a random arrangement of a finite set of elements such that each possible arrangement is equally likely. Efficient random permutation generation has attracted considerable attention due to its wide range of practical applications.  One example is the technique of randomly hopping frequencies between cellular phones and their base stations, which makes it extremely difficult to eavesdrop on a connection~\cite{ristic2022}. Another area where random permutations play a fundamental role is in the security of many cryptographic systems, where they are used to construct block ciphers that are resistant to various types of attacks~\cite{czumaj2015random}. 
Randomized algorithms also employ random permutations to ensure that their behavior cannot be predicted based on a specific input order~\cite{motwani95}. 
Random permutations also find widespread applications across diverse fields, including spatial data analysis~\cite{Computation_free_nonparametric_spatial}, network routing and traffic management~\cite{iwankowicz2021permutation}, the analysis of DNA and RNA sequences in biology~\cite{forsdyke1994shuffling}, among others.

Classically, a uniform random permutation can be generated by the Fisher--Yates shuffle \cite{fisher1963statistical}, which samples each of the \(n!\) permutations with equal probability by realizing \(n(n-1)/2\) transpositions. By contrast, the information-theoretic minimum number of \emph{binary routing decisions} needed to distinguish among all \(n!\) possible permutations is only \(\log_2(n!)=\Theta(n \log n)\)~\cite{cormen2022introduction}. This gap motivates the use of structured routing architectures. In particular, multistage networks such as the Bene\v{s} network can realize any prescribed permutation using \(n(\log_2 n - 1/2)\) binary \(2\times2\) switches, while any universal architecture built from such binary switches must contain at least \(\log_2(n!)\) elements. Bene\v{s}-type networks are therefore asymptotically optimal for universal permutation routing. Uniform generation, however, is a strictly stronger requirement than universal routability: a routing network guarantees that every permutation can be realized, but it does not by itself ensure that local switch choices induce the uniform distribution over \(S_n\). 

Because permutation symmetry appears in many parts of quantum information science, the problem of translating random permutations into the quantum setting has attracted considerable attention. Quantum versions of permutation generation arise in areas as diverse as quantum error correction~\cite{ouyang2026theoryquantumerrorcorrection}, quantum cryptography~\cite{Renner2007Symmetry}, and quantum optimization~\cite{10.1145/780542.780546}.

An explicit circuit-level construction for preparing a uniform superposition over all permutations was first given by Barenco \emph{et al.}~\cite{barenco1997stabilization}. Their method uses coherent controlled-SWAP operations and may be viewed as a quantum counterpart of a classical switching network: instead of random switch settings, quantum interference creates a superposition of all reorderings with identical amplitudes. Since then, several alternative frameworks have been developed, motivated by different goals such as clearer structure, reduced resource overhead, or better suitability for particular algorithmic settings. Representative examples include mixed-radix encodings of permutations~\cite{chiew2019graph}, constructions based on Lehmer codes~\cite{Marsh2020}, a general formalization of permutation oracles~\cite{10.1145/3717823.3718266}, quantum variants of the Steinhaus--Johnson--Trotter procedure for permutation sampling~\cite{11089489}, and insertion-based methods designed to encode dynamic-programming formulations of combinatorial optimization problems directly into permutation superpositions~\cite{BAI2025115423}. More recently, Binkowski \emph{et al.}~\cite{binkowski2025quantum} introduced a quantum version of the classical Fisher--Yates shuffle, leading to a family of algorithms for preparing coherent permutation superpositions on composite registers.

These contributions highlight the richness of existing approaches to permutation generation in quantum circuits; however, they also reveal a common structural limitation. In fact, every exact uniform construction currently known assumes full connectivity. In practice, however, near-term quantum devices do not permit arbitrary long-range interactions, and the resulting topology constraints fundamentally alter the permutation-generation problem. An important question is no longer simply whether every permutation is realizable, but whether local gate choices can be arranged so that the induced distribution over all \(n!\) permutations is exactly uniform. Such a problem is inherently combinatorial and strongly dependent on the underlying interaction graph. 

This observation motivates the introduction of a variational quantum circuit (VQC), by which we mean a parameterized quantum circuit with tunable gate parameters optimized classically to generate permutations uniformly. Parameterized quantum circuits are the core quantum component of variational quantum algorithms \cite{Cerezo2021,mcclean2016theory}, and have been used in applications ranging from ground-state estimation in quantum chemistry \cite{peruzzo2014variational,yordanov2020efficient} to combinatorial optimization, most prominently through the Quantum Approximate Optimization Algorithm (QAOA) \cite{farhi2014quantum,perez2024variational}. Their appeal lies in the hybrid quantum--classical paradigm: a shallow, hardware-compatible quantum circuit defines a family of candidate transformations, while a classical optimizer adjusts the circuit parameters to minimize a prescribed cost function. In the present setting, this flexibility allows the circuit architecture to respect local connectivity constraints while optimizing for exact uniformity over permutations.

In this work, we study uniform permutation generation within a variational quantum circuit framework. We develop a general method for topology-constrained circuit design in which the circuit architecture is determined by the available qubit connectivity encoded by the underlying interaction graph, while the variational parameters are optimized to enforce the target permutation statistics. Unlike previous approaches, our analysis does not assume fully connected architectures, an important distinction since all-to-all connectivity is rarely available on realistic quantum hardware.

We specialize this framework for linear nearest-neighbor topologies, where subsystems are arranged in a chain, and only adjacent swaps are permitted. In this setting, we construct a circuit that achieves exact uniformity with linear depth, \(O(n)\), thereby improving the tightest known quadratic depth and all-to-all connectivity. This yields a linear-factor improvement over \(O(n^2)\)-depth constructions. Finally, we prove that a quantum Bene\v{s}-like architecture is intrinsically non-uniform: despite its depth \(O(\log_2 n)\), it cannot generate the uniform distribution over permutations for any choice of variational parameters.

\section{Uniform Quantum Permutation}
In this work, we consider a composite quantum system composed of $n$ subsystems (e.g., qubits), each associated with a local Hilbert space ${\cal H}$. The global Hilbert space is therefore given by ${\cal H}^{\otimes n}$. A pure configuration of the system can be written as a product state $\ket{\psi} = \ket{\psi_1} \otimes \ket{\psi_2} \otimes \dotsb \otimes \ket{\psi_n}$,
where $\ket{\psi_i} \in {\cal H}$ denotes the state of the $i$-th subsystem.

To formalize the action of permutations on such product states, we define, for each $\pi \in S_n$, a unitary operator $U_\pi$ acting on ${\cal H}^{\otimes n}$ as 
\begin{equation}
    U_\pi \ket{\psi_1, \psi_2, \dotsc, \psi_n}
= \ket{\psi_{\pi^{-1}_{1}}, \psi_{\pi^{-1}_2}, \dotsc, \psi_{\pi^{-1}_n}}.
\end{equation}
The set $\{ U_\pi : \pi \in S_n \}$ constitutes a unitary representation of the symmetric group $S_n$, satisfying \( U_{\pi_1} U_{\pi_2} = U_{\pi_1 \pi_2}\) and \(U_\pi^\dagger = U_{\pi^{-1}} \).
This representation captures all possible reorderings of the $n$ subsystems within the composite quantum system.  We consider the task of coherently generating an equal-weight superposition over all permutations of a product state.

\begin{definition} \label{def:Unitaryuniform}
Let \(\ket{\psi}=\bigotimes_{i=1}^n \ket{\psi_i}\) be a product state over \(n\) subsystems, where \(\ket{\psi_i}\in\mathcal{H}_i\), and let \(\ket{0}_a=\bigotimes_{j=1}^m \ket{0}_j \in \mathcal{H}_a\) denote a fixed ancillary state on \(m\) auxiliary subsystems. Also, consider a unitary operator \(\hat{\mathcal{U}}\) acting on joint Hilbert space. We say that \(\hat{\mathcal{U}}\) generates a uniform superposition over all permutations of \(\ket{\psi}\) if
\begin{equation}\label{eq:def1}
\hat{\mathcal{U}}\ket{\psi}\otimes\ket{0}_a
=
\frac{1}{\sqrt{n!}}
\sum_{\pi\in S_n}
\ket{\psi_{\pi^{-1}_1},\ldots,\psi_{\pi^{-1}_n}}\otimes \ket{\varphi_\pi}_a
\end{equation}
where \(S_n\) is the symmetric group on \(n\) elements and \(\ket{\varphi_\pi}_a\in\mathcal{H}_a\) is the ancillary state associated with permutation \(\pi\).
\end{definition}

Here, we propose quantum circuits for generating a uniform superposition of all permutation in the symmetric group \(S_n\) using swap operations (transposition) between subsystem.  

From Definition~\ref{def:Unitaryuniform}, one also finds the completely positive trace-preserving (CPTP) map \(\mathcal{M}(\hat{\rho})=\mathrm{Tr}_a\left(\hat{\mathcal{U}}\hat{\rho}_{0}\hat{\mathcal{U}}^\dagger\right)\) with \(\hat{\rho}_{0}=\ket{\psi}\bra{\psi}\otimes\ket{0}_a\bra{0}_a\) by tracing over the auxiliary subsystems. This yields 
\begin{equation}\label{eq: G quantum consensus}
     \mathcal{M}(\hat{\rho})= \frac{1}{n!} \sum_{\pi \in S_n} \hat{U}_\pi \, \hat{\rho}_0 \, \hat{U}_\pi^\dagger,
\end{equation}
where \(\hat{U}_\pi\) implements the permutation \(\pi\) in the symmetric group \(S_n\).

The complexity of the proposed circuit is evaluated using three resource measures: circuit size, circuit depth, and ancillary-qubit count. The \textit{circuit size} is the total number of controlled-SWAP gates, considered the elementary operation in the permutation circuits studied here. Each controlled-SWAP can be decomposed into two CNOT gates and one Toffoli gate, and a Toffoli gate itself admits a standard decomposition into six CNOT gates and nine single-qubit gates. Thus, this cost model remains compatible with conventional gate-counting frameworks. The \textit{circuit depth} is the number of sequential layers required to execute the circuit under the assumption that gates acting on disjoint qubits may be applied in parallel. Finally, the \textit{ancillary-qubit count} is the number of auxiliary qubits required to control the swap operations acting on the input subsystems.

To implement arbitrary unitary transformations within hardware gate-set and connectivity constraints, we use the standard qubit-register encoding of finite-dimensional subsystems~\cite{Barenco1995Elementary,nielsen2010quantum}.

\section{Qubit-level decomposition of subsystem operations}

A standard approach to implementing arbitrary unitary operations on composite quantum systems is to express them in terms of elementary gates acting on qubits, in a manner compatible with the gate sets and connectivity constraints of current hardware platforms~\cite{Barenco1995Elementary,nielsen2010quantum}. 

In this representation, each \(d\)-dimensional subsystem is encoded into a register of \(\lceil \log_2 d \rceil\) qubits, which is the minimum number of qubits needed to represent its computational basis states in binary form. This encoding provides a concrete bridge between abstract subsystem-level operations and their realization at the physical qubit level. Under this decomposition, operations between subsystems are implemented register-wise. In particular, swapping two \(d\)-dimensional subsystems requires \(\lceil \log_2 d \rceil\) SWAP gates, applied pairwise between the corresponding qubits of the two registers. Likewise, a controlled-SWAP between two subsystems is implemented by conditioning each of these \(\lceil \log_2 d \rceil\) constituent SWAP gates on the same ancillary control qubit.

To simplify the presentation, we will henceforth depict each encoded register as a single subsystem wire. Accordingly, subsystem-level SWAP and controlled-SWAP operations should be understood as shorthand for the corresponding qubit-level implementations described above.

\begin{figure*}
    \centering
    \includegraphics[width=0.99\linewidth]{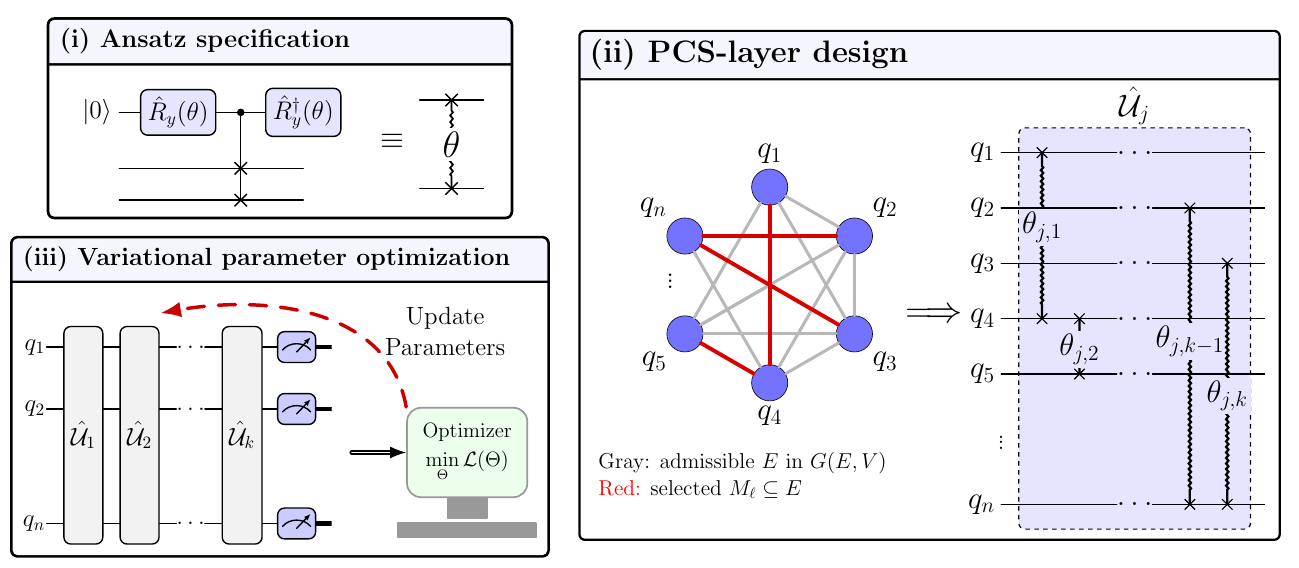}
    \caption{\textbf{Three-stage variational quantum-circuit construction for uniform permutation generation.} (i) Each candidate swap \(\hat U_{ij}\) is realized by a parameterized controlled-SWAP (PCS) gate, with the control qubit prepared by \(\hat R_y(\theta)\).  (ii) The PCS-layer architecture is determined by an interaction graph \(G=(V,E)\); at each layer \(\ell\), an ordered subset \(M_\ell\subseteq E\) is chosen to specify the active PCS operations. Gray edges represent admissible interactions, and red edges the selected ones.  (iii) For the fixed layered architecture, the variational parameters are optimized by minimizing \(\mathcal{L}(\Theta)\), so that the resulting channel approximates the target uniform permutation map.}
    \label{fig:VAC}
\end{figure*}

\section{Variational quantum circuit}

A variational circuit consists of a fixed gate sequence (an ansatz) with tunable continuous parameters, which are optimized to minimize a cost function quantifying the deviation from the target uniform permutation generator. The variational proceed in the following three stages:

\textbf{(i) Ansatz specification:} The circuit employs \(m\) ancillary qubits, each controlling a potential swap \(\hat{U}_{ij}\). Each such potential swap \(\hat{U}_{ij}\) is represented by a parameterized controlled-SWAP (PCS) operator. This gate generalizes the standard controlled-SWAP by allowing the control qubit to be initialized through a rotation about the Pauli-\(Y\) axis. Specifically, the control qubit is prepared using \(\hat{R}_y(\theta)=e^{-i\theta \hat{Y}/2}\), so that
\(\hat{R}_y(\theta)\ket{0}
=\cos\!\left(\frac{\theta}{2}\right)\ket{0}+\sin\!\left(\frac{\theta}{2}\right)\ket{1}.
\)
The parameter \(\theta\) determines the amplitudes of the two branches and therefore the probability that the swap is activated. In particular, the probability that the target subsystems remain unchanged is \(\cos^2\!\left(\frac{\theta}{2}\right)\). 

Figure~\ref{fig:VAC}~(i) illustrates this operator. These amplitudes are treated as free parameters of the variational ansatz and are collected in the parameter vector \(\Theta = (\theta_1, \dots, \theta_m)\).

\textbf{(ii) PCS-layer design:} In the second step, we design the layered PCS architecture, as illustrated in Figure~\ref{fig:VAC}~(ii).  The admissible PCS operations are specified by an interaction graph \(G=(V,E)\), where each vertex in \(V\) represents a target subsystem and each edge in \(E\) indicates that a direct PCS operation can be applied between the corresponding pair. The graph \(G\) therefore determines the circuit connectivity.

For each layer \(\ell\), we select a collection of admissible edges from \(E\), denoted by \(M_\ell\) as an ordered list of selected admissible edges, since the corresponding operations need not commute in general. PCS operations are then applied on the selected pairs in this prescribed order.

A minimal requirement for the variational ansatz is that the layered PCS architecture be able to realize every permutation of the \(n\) target subsystems, i.e., every element of the symmetric group \(S_n\). Equivalently, for every \(\pi \in S_n\), the ansatz should have support on the corresponding permutation unitary \(\hat U_\pi\).

To ensure this, we assume that each active PCS operation has nonzero support on both the identity branch (no swap) and the transposition branch associated with its edge. In our parametrization, this is ensured by choosing \(\theta \in (0,\pi)\), so that neither branch is excluded. For a given layer \(\ell\), the set of permutations generated by that layer is therefore
\begin{equation}
\mathcal{S}_\ell=\left\{\prod_{(i,j)\in M_\ell} \sigma_{ij}\;\middle|\;\sigma_{ij}\in\{\hat{I},\hat{U}_{i\,j}\}\right\}\subseteq S_n
\end{equation}
where \(\hat{I}\) denotes the identity permutation and \(\tau(e_{\ell,r})\) denotes the transposition associated with the selected edge \(e_{\ell,r}\). Thus, \(\mathcal{S}_\ell\) is the set of all permutations that can be generated by the PCS block in layer \(\ell\). The corresponding set of permutation unitaries is \(\{\hat U_\sigma \mid \sigma \in \mathcal{S}_\ell\}\).

To characterize the permutations reachable after multiple layers, we introduce the set product. 

\begin{definition}[Set product]\label{def:set_product}
For subsets \(A,B \subseteq S_n\), define
\begin{equation}
A\cdot B=\{\, a\circ b \mid a\in A,\; b\in B\,\},
\end{equation}
where \(a\circ b\) denotes composition, with \(b\) applied first and \(a\) applied second.
\end{definition}

The set of permutations reachable after \(k\) layers is then \(\mathcal{G}_k=\mathcal{S}_k \cdot \mathcal{S}_{k-1}\cdots \mathcal{S}_1\). The architectural design criterion is that there exists \(k\) layers such that \(\mathcal{G}_k = S_n\). In other words, the sequence of selected edge sets \(\{M_\ell\}_{\ell=1}^k\) must be chosen such that successive compositions of the corresponding PCS layers generate the full symmetric group. Once this structural condition is satisfied, the variational parameters can be used to control the induced distribution over permutations.

The architectural design criterion is therefore to choose the sequence \(\{M_\ell\}_{\ell=1}^k\) such that \(\mathcal{G}_k=S_n\). In principle, the layer structure \(\{M_\ell\}_{\ell=1}^k\) could itself be treated as a variational object, with the VQC searching for an architecture such that \(\mathcal{G}_k=S_n\). This is not our aim here. Instead, we choose PCS-layer designs that already generate the full symmetric group by construction. The subsequent optimization step is then used only to shape the induced distribution over \(S_n\), ideally toward the uniform distribution, rather than to discover a permutation-generating architecture.

\textbf{(iii) Cost-function minimization:} Once the PCS-layer design is fixed, the remaining task is to optimize the variational parameters \(\Theta\) so that the induced channel approximates the target uniform permutation channel \(\mathcal{M}\) in Eq.~\ref{eq: G quantum consensus}. For each \(\Theta\), the circuit defines a channel \(\mathcal{E}_\Theta\) through the prescribed PCS architecture on \(G=(V,E)\), followed by tracing out the ancillary degrees of freedom. We therefore define the cost function
\begin{equation}\label{eq:costfunction}
    \mathcal{L}(\Theta)
    =
    \frac{1}{2}
    \sup_{\hat{\rho}_0}
    \left\|
         \mathcal{E}_\Theta(\hat{\rho}_0)
         -
         \mathcal{M}(\hat{\rho}_0)
    \right\|_1,
\end{equation}
Here, \(\|\cdot\|_1\) is the trace norm, and the supremum over \(\hat{\rho}_0\) captures the worst-case deviation between the two channels on input states. Hence, the optimization does not search for a circuit architecture that generates permutations; that structural requirement has already been enforced at the design stage. Instead, it adjusts the parameters within the fixed architecture so that the induced distribution over permutations approaches the uniform one (see Figure~\ref{fig:VAC}~(iii)). An optimal parameter set \(\Theta^*\) is obtained by minimizing \(\mathcal{L}(\Theta)\). If the minimum satisfies \(\mathcal{L}(\Theta^*)=0\), then \(\mathcal{E}_{\Theta^*}=\mathcal{M}\), i.e., the variational circuit exactly realizes the target channel.

In order to apply the cost-function minimization approach to uniform quantum permutation generation, we replace the cost function of Eq.~\ref{eq:costfunction} by a state-independent projection criterion. The symmetric projector associated with the permutation action is \(\Pi_{\rm sym}=\frac{1}{n!}\sum_{\pi\in S_n}\hat{U}_\pi\), where \(\hat{U}_\pi\) denotes the representation of the permutation \(\pi\) on the relevant state space. Denoting by \(P(\Theta)\) the PCS-induced operator in the same representation, we minimize
\begin{equation}
    \mathcal{L}_{\rm proj}(\Theta)
    =
    \|\Pi_{\rm sym}-P(\Theta)\|_F^2 .
\end{equation}
Here, \(\|\cdot\|_F\) denotes the Frobenius norm. Since both operators act on the same representation space, \(\mathcal{L}_{\rm proj}=0\) is equivalent to \(P(\Theta)=\Pi_{\rm sym}\), and hence to exact realization of the target symmetric projection. Conversely, \(\mathcal{L}_{\rm proj}>0\) implies that the PCS-induced map cannot coincide with the target projection on all inputs. Thus, this loss provides a computationally efficient surrogate whose zeros coincide with exact realization of the target map, although its magnitude is not the worst-case trace-distance channel error.

\section{Linear topology}
In this section, we apply the variational framework based on a linear topology and evaluate the uniformity of existing topologies. The first constraint is that the PCS-layer design set should be disjoint, i.e., the selected edges in the graph should be a matching that is admissible from the graph connectivity. This will reduce the depth quantum circuit to $k$ layers.

The simplest graph connectivity is the linear topology. In this setting, each subsystem \(i\) is allowed to exchange its quantum state exclusively with its nearest neighbors, \(i-1\) and \(i+1\) (whenever they exist). The resulting architecture forms a one-dimensional chain that restricts interactions to strictly local, nearest-neighbor exchanges between adjacent subsystems. While this may seem restrictive, a key feature of this construction is its applicability to any connected graph that can embed a linear ordering of subsystems (via a Hamiltonian path or other linear arrangements~\cite{bondy2008graph}). This is particularly relevant given that most current quantum architectures admit Hamiltonian paths, thereby enabling the use of quantum circuits designed for the linear topology on such platforms. Among the architectures admitting Hamiltonian paths are lattice or grid structures (as in IBM’s Eagle chip), all-to-all topologies (used in trapped-ion and photonic implementations), heavy-hex/hexagonal lattices (employed in IBM’s Falcon and Hummingbird chips), and linear chains (common in trapped-ion devices), among others.

We describe the step-by-step construction of the linear quantum circuit in the followign steps. First, the control register is prepared by applying independent rotations about the Pauli-\(Y\) axis on each of the \(n(n-1)/2\) ancillary qubits that are used, in a specific pattern that produces a uniform superposition.

\begin{figure*}
    \centering
    \includegraphics[width=0.95\linewidth]{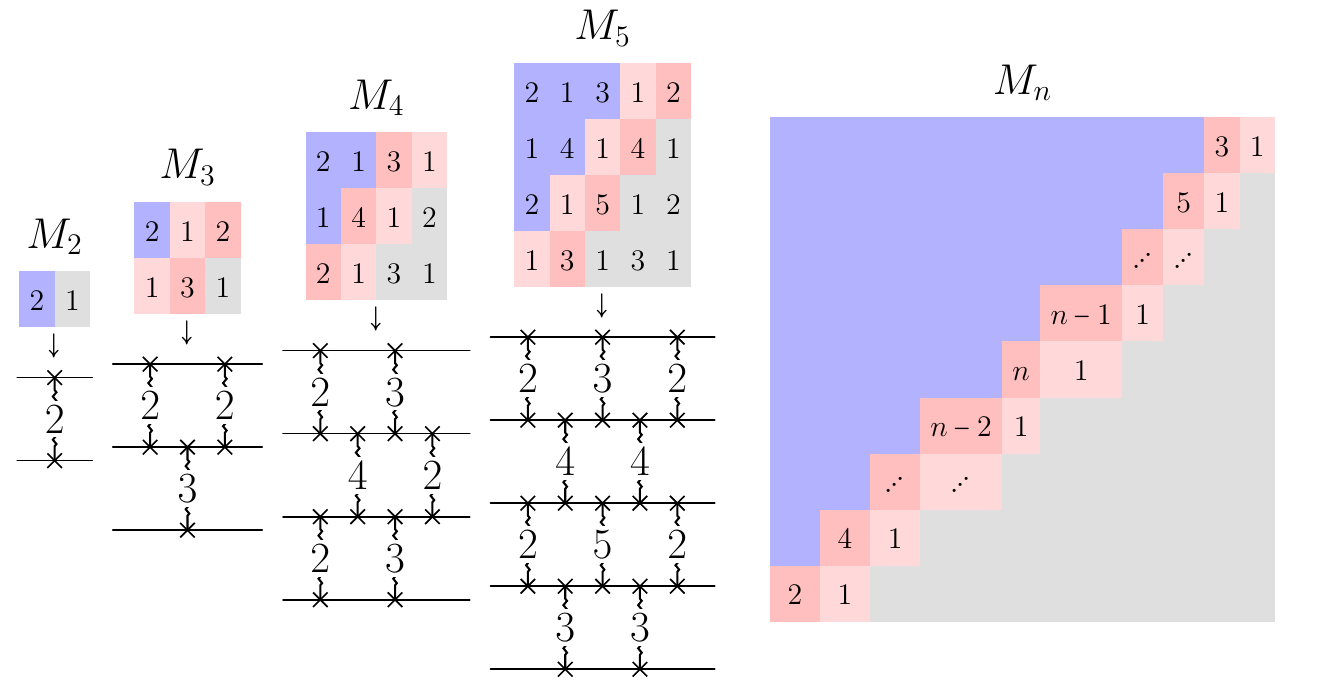}
\caption{\textbf{Recursive construction of the non-swap probability matrices \(M_n\) and their corresponding linear-topology circuits.} (\textbf{Top}) Examples for \(n=2,3,4,5\) are shown. Each \(M_n\) is obtained from the upper and lower blocks inherited from \(M_{n-1}\), shown in blue and gray, together with two newly added anti-diagonals, shown in red. For even \(n\), the numbered anti-diagonal starts with even entries \(2,\ldots,n\) followed by odd entries \(n-1,\ldots,3\), while the other anti-diagonal contains only \(1\)s. For odd \(n\), the roles of the two anti-diagonals are interchanged. (\textbf{Bottom}) The associated quantum circuits implementing the same recursive linear topology.}
    \label{fig:Inductive}
\end{figure*}

Second, the circuit applies a sequence of  \(n\) layers (except for the case $n=2$ that has only $1$ layer, see Figure \ref{fig:Inductive}), arranged to emulate the \emph{compact} linear topology illustrated in Figure \ref{fig:Inductive}. In this compact linear topology, each layer applies \(\lfloor n/2 \rfloor\) controlled-swap gates between adjacent subsystems. Odd layers act on pairs \((1,2), (3,4), \ldots, (2\lfloor \frac{n}{2} \rfloor-1,2\lfloor \frac{n}{2} \rfloor)\), while even layers are shifted by one position and act on pairs \((2,3), (4,5), \ldots, (2\lfloor \frac{n-1}{2}\rfloor,2\lfloor \frac{n-1}{2}\rfloor+1)\). The complete circuit alternates these layers, resulting in an overall circuit depth of \(n\).

Before addressing uniformity optimization and the associated variational analysis, we first establish a structural result: the minimal circuit depth required for a linear topology to generate all permutations in the symmetric group \(S_n\). This is stated in the following lemma, whose the proof has been deferred to Appendix~\ref{appenidx1}.

\begin{lemma}
\label{lemma:permutationI}
Let \(S_n\) denote the symmetric group on \(n\) elements. Define the two sets of disjoint adjacent-swap layers 
\begin{align}
E_1&=\left\{
\prod_{j=1}^{\lfloor \frac{n}{2}\rfloor} W_j
\;\middle|\;W_j\in \{\hat{I},\hat U_{2j-1,2j}\}
\right\},\\ 
E_2 &=\left\{\prod_{j=1}^{\lfloor \frac{n-1}{2} \rfloor} W_j
\;\middle|\; W_j\in \{\hat{I},\hat U_{2j,2j+1}\} \right\}.
\end{align}
Then every permutation \(\pi\in S_n\) can be written as an element of \((E_1\cdot E_2)^k,\) with \(k=\frac n2 \) when \(n\) is even, and as an element of  \((E_1\cdot E_2)^k\cdot E_1\) with \(k= \frac{n-1}{2} \). when \(n\) is odd.
\end{lemma}

This lemma establishes that a depth-\(n\) circuit on a linear topology is capable of generating the full symmetric group \(S_n\). 

\begin{figure}[t!]
    \centering
    \includegraphics[width=0.7\linewidth]{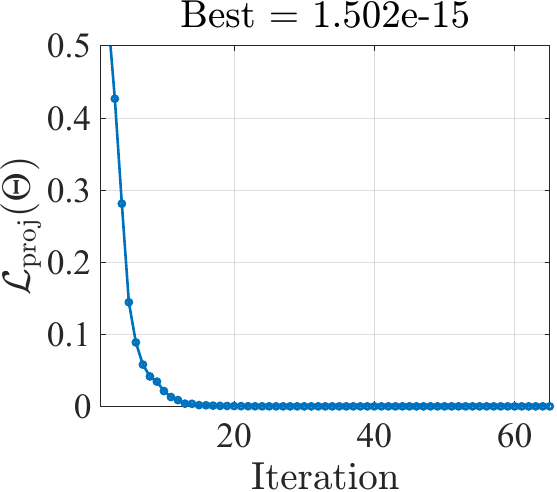}
   \caption{Projection-loss optimization for the line-topology PCS ansatz on \(S_8\). The best run over multiple random initializations reaches a residual of order \(10^{-16}\).}
\label{fig:linearopt}
\end{figure}

We now turn to the variational optimization of the PCS parameters toward the uniform permutation channel. As shown in Fig.~\ref{fig:linearopt} for the specific case of \(S_8\)---but we remark that we have observed the same in all tested cases--- optimization of the projection loss for the swap sequence associated with the line topology converges to numerical precision. The plotted curve corresponds to the best run over multiple random initializations, with a final residual of order \(10^{-16}\). This demonstrates that the line-topology PCS ansatz realizes the target symmetric projector up to numerical precision. Consequently, the resulting circuit implements the uniform permutation average over \(S_8\).

Remarkably, the optimized solutions exhibit a simple inductive structure. More precisely, each PCS gate acting on neighboring subsystems \((i,j)\) is assigned swap probability \(p_{ij}=\cos^2\!\left(\frac{\theta_{ij}}{2}\right)=1-\frac{1}{m_{ij}}\), where the coefficients \(m_{ij}\) are arranged in an \((n-1)\times n\) matrix \(M_n\)  whose elements are exactly integer numbers between $1$ and $n$. We remark that a $1$ in the matrix represents a zero-probability swap, hence a connection between subsystems which is not used in a specific layer of the controlled-SWAP operations.
The corresponding non-swap probability that emerges from optimization is therefore \(1/m_{ij}\). 
Each entry of \(M_n\) therefore specifies the reciprocal of the non-swap probability associated with a nearest-neighbor PCS gate. The structure of \(M_n\) shows that longest swap paths (the diagonal ones) happen with probability $\frac{1}{n!}$ obtained as the product of all the reciprocals of the factors in $n!$. However, for permutations that can be obtained as a superposition of multiple alternative swap combinations, the result is neither obvious nor intuitive, still we observe that it holds.

The observation of numerical optimization results leads to the identification of a recurrent structure in the matrix $M_n$, which allows to build it recursively starting from $n=1$. The construction starts with the base matrix
\(M_2 = [\,2\;\;1\,]\) with two subsystems, where each entry $x$ represents a swap probability of $1-1/x$ (i.e., the entry \(1\) denotes the no-swap choice). For \(n>2\), the matrix \(M_n\) is obtained recursively by combining two copies of the \((n-1)\)-system pattern, shown as the upper and lower blocks in Fig.~\ref{fig:Inductive}, with two additional anti-diagonals of length \(n\) running from lower left to upper right. For \(n=3\), the construction starts from \(M_2=[\,2\;\;1\,]\), and two new
anti-diagonals are inserted. The left anti-diagonal consists of entries equal to \(1\), while the right anti-diagonal contains the sequence \(3,2\), read
from bottom to top. Thus, the new entries added at this step are the anti-diagonal pairs \((1,3)\) and \((1,2)\), as shown in red in Fig.~\ref{fig:Inductive}. For \(n=4\), the construction uses the pattern for \(n=3\) as the inherited upper and lower blocks. Since \(n=4\) is even, the left anti-diagonal contains
the sequence \(2,4,3\), while the right anti-diagonal consists of entries equal to \(1\). Therefore, the red anti-diagonal pairs added at this step are
\((2,1)\), \((4,1)\), and \((3,1)\), placed along anti-diagonals running from lower left to upper right.

The entries on these two new anti-diagonals depend on the parity of \(n\). If \(n\) is even, the left anti-diagonal contains the sequence \([2,4,\ldots,n,\; n-1, n-3,\ldots,3]\), whereas the right anti-diagonal
consists only of entries equal to \(1\). If \(n\) is odd, the roles of the two anti-diagonals are interchanged: the left anti-diagonal consists only of
entries equal to \(1\), while the right anti-diagonal contains \([3,5,\ldots,n,\; n-1, n-3,\ldots,2]\). Thus, each step from \(M_{n-1}\) to \(M_n\) preserves the previously constructed blocks and inserts a new pair of parity-dependent anti-diagonals. This recursive pattern is illustrated for several system sizes in Fig.~\ref{fig:Inductive}.

The resulting variational solution appears to reproduce the uniform permutation channel exactly. We summarize the numerical evidence as follows.

\begin{observation}[Empirical correctness]
For all tested system sizes \(n\leq 12\), the compact linear circuit defined by the inductively constructed swap matrix \(M_n\) reproduces the uniform permutation channel \(\mathcal{M}\) to numerical precision, with the observed trace-distance cost in Eq.~\ref{eq:costfunction} in the order of \(10^{-15}\).
\end{observation}

Although we do not presently have a formal proof, the numerical agreement with the target channel holds to machine precision in all tested cases. This motivates the following conjecture.

\begin{conjecture}
For every \(n\), the inductively constructed swap matrix \(M_n\) defines the no-swap probabilities of a quantum circuit on a linear nearest-neighbor topology that generates a uniform coherent superposition over all permutations of the \(n\) subsystems. Moreover, the resulting circuit uses exactly \(\frac{n(n-1)}{2}\) controlled-SWAP operations, has circuit depth \(n\), and requires \(\frac{n(n-1)}{2}\) ancillary control qubits.
\end{conjecture}

The depth of the quantum circuit described above is linear in \(n\), and it improves on the depth of the best-known quantum circuits for coherent uniform permutation generation \cite{Barenco1995Elementary}. Apart from the uniformity requirement, it has been shown that a symmetric-group generator must use at least \(\Omega(n\log_2 n)\) swaps, which translates into a circuit-depth lower bound of \(\Omega(\log_2 n)\) under parallel implementation. In the following section we show that, although circuit implementations of such symmetric-group generators exist, exact uniformity is impossible to achieve in the corresponding architecture.

\section{Beneš-like Quantum Circuit}

The Beneš network is an optimal rearrangeable non-blocking architecture that realizes all \(n!\) permutations using a logarithmic number of switching stages, \(2\log_2 n - 1\), and \(N\log_2 n - \frac{n}{2}\) switches \(2\times 2\) (we assume that \(n\) is a power of two). Classical randomized switching networks assign each \(2\times 2\) switch an independent fair control bit, thereby inducing a probability distribution over permutations. For the Beneš network, this induced distribution was studied by Gelman and Ta-Shma, who showed that it is \((\frac{q(q-1)}{2n})\)-almost \(q\)-set-wise independent, rather than exactly uniform over \(S_n\)~\cite{gelman2014benes}.

In any fixed switching network with \(K\) unbiased binary switches, the probability of an output permutation \(\pi\) has the form \(\Pr[\pi] = M_\pi/2^K\), where \(M_\pi\) is the number of switch configurations that realize \(\pi\). Exact uniform generation over \(S_n\) would require \(\Pr[\pi] = M_\pi/2^K\) for every \(\pi\), which is impossible for \(n\geq 3\) because \(n!\) does not
divide \(2^K\).

Motivated by the efficiency of the Beneš network, we construct a \emph{Beneš-like} quantum circuit by replacing each classical switch with a controlled-SWAP gate whose activation is mediated by a control qubit. This substitution yields coherent superpositions over permutations. However, the same counting argument shows that exact uniform coherent permutation generation cannot be achieved by simply assigning an unbiased control qubit to each
controlled-SWAP gate.

We then generalize the design by introducing the tunable parameters \(\Theta\) and employ a variational framework to test whether exact uniformity can be achieved by tuning swapping probabilities.

\begin{figure}
    \centering
    \includegraphics[width=\linewidth]{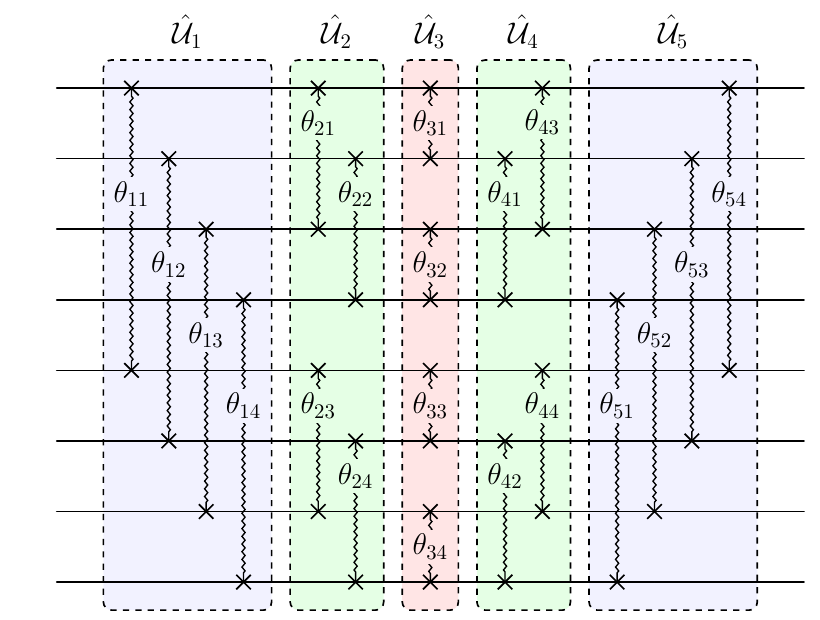}
   \caption{Quantum Beneš circuit for \(n=8\).}
    \label{fig:Benes}
\end{figure}

As in the previous quantum circuit, the control register is prepared by applying independent rotations about the Pauli-\(Y\) axis to each of the \(n\log_2 n - \frac{n}{2}\) ancillary qubits. These rotations parameterize variable superposition states that serve as free parameters for the variational optimization.

As in the previous quantum circuit, the control register is prepared by applying independent rotations about the Pauli-\(Y\) axis to each of the \(n\log_2 n - \frac{n}{2}\) ancillary qubits. These rotations parameterize variable superposition states that serve as free parameters in the variational optimization.

Analogously to its classical counterpart, the quantum Beneš circuit operates in two symmetric halves corresponding to the forward and reverse stages of the network. In the \emph{half-split} (or far-node) construction, the controlled-SWAP gates are first applied between subsystem pairs that are
maximally separated and are then progressively shifted toward nearest neighbors. In general, for \(n\) subsystems (assuming that \(n\) is a power of
two), the circuit applies a sequence of \(\log_2 n\) layers indexed by \(k = 1,2,\ldots,\log_2 n\). At layer \(k\), the controlled-SWAP gates act on all pairs \((i, i + d_k)\), where \(d_k = \frac{n}{2^{k}}\). After reaching the nearest-neighbor layer, the sequence of layers is executed in reverse order to restore the initial spatial arrangement, giving a total of \(2\log_2 n - 1\) layers. 

For example, when \(n = 4\), the first layer connects the pairs \(\{(1,3),(2,4)\}\), corresponding to a separation distance of \(d_1 = \frac{n}{2} = 2\). The second layer connects the pairs \(\{(1,2),(3,4)\}\), and the final layer again connects \(\{(1,3),(2,4)\}\). For \(n = 4\), the circuit generates all permutations in the symmetric group \(S_4\) with uniform probability. An admissible set of PCS probabilities is
\begin{equation}
    \Theta =
\left\{
\frac{1}{2},\frac{1}{2},
\frac{3-\sqrt{3}}{6},
1-\frac{3-\sqrt{3}}{6},
\frac{1}{2},\frac{1}{2}
\right\}.
\end{equation}
The total number of controlled-SWAP operations is six, which matches the count \(\frac{n(n-1)}{2}\) used in the preceding quantum circuits considered in this
work.

\begin{figure}
    \centering
    \includegraphics[width=\linewidth]{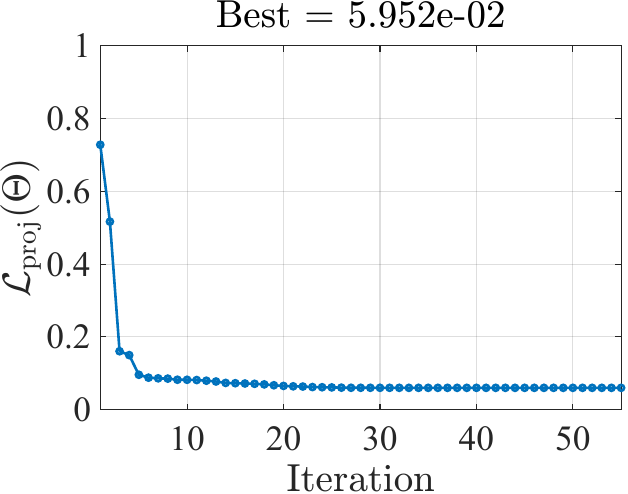}
   \caption{Optimization of the projection loss \(\mathcal{L}_{\rm proj}\) for the selected PCS swap sequence. The curve shows the best run over multiple random initializations. The finite residual indicates that the ansatz cannot exactly realize the symmetric projector, and therefore does not implement the uniform permutation average over \(S_8\).}
    \label{fig:benesopt}
\end{figure}

However, for \(n = 8\), the circuit contains \(20\) controlled-SWAP operations, and hence \(20\) PCS parameters. The first layer connects the pairs \(\{(1,5),(2,6),(3,7),(4,8)\}\), corresponding to a separation distance of \(d_1 = \frac{n}{2} = 4\). The second layer connects the pairs \(\{(1,3),(2,4),(3,5),(4,6),(5,7),(6,8)\}\) with distance \(d_2 = \frac{n}{4} = 2\). The third layer performs nearest-neighbor swaps \(\{(1,2),(3,4),(5,6),(7,8)\}\), with distance \(d_3 = \frac{n}{8} = 1\), as illustrated in Fig.~\ref{fig:Benes}. Although the circuit generates all permutations in the symmetric group \(S_8\), no admissible set of parameters yields a uniform distribution.  As shown in Fig.~\ref{fig:benesopt}, optimization of the projection loss for the swap sequence of the quantum Bene\v{s} circuit in Fig.~\ref{fig:Benes} converges to a finite residual. The plotted curve corresponds to the best run over multiple random initializations. Since the loss remains bounded away from numerical zero, the quantum Bene\v{s} circuit does not exactly realize the target symmetric projector. Consequently, the resulting circuit does not implement the uniform permutation average over \(S_8\). This leads to the following conjecture.

\begin{conjecture}
For the quantum Bene\v{s} circuit with \(n=8\), no choice of variational parameters \(\Theta\) yields the uniform distribution over \(S_8\).
\end{conjecture}

If confirmed, this fact shows that universal permutation realizability does not, by itself, imply exact uniform permutation generation. Although the quantum Bene\v{s} architecture can realize every permutation and has logarithmic depth, its variational degrees of freedom are insufficient to induce the uniform distribution. This naturally motivates the following conjecture.

\begin{conjecture}
Under locality-constrained controlled-SWAP architectures, the currently known upper bounds for exact uniform permutation generation are asymptotically optimal. In particular, any exact uniform permutation generator requires circuit size \(\Omega(n^2)\) and circuit depth \(\Omega(n)\).
\end{conjecture}

If true, this conjecture would establish that exact uniform permutation generation is asymptotically more demanding than mere permutation realizability. At present, however, it remains open whether such lower bounds can be proved, or whether a different circuit architecture beyond the Bene\v{s} family can achieve exact uniformity with complexity strictly between the current \(O(n^2)\)-size, \(O(n)\)-depth upper bounds and the information-theoretic limits for universal realizability.

\section{Conclusion}

In this work, we studied uniform permutation generation within a variational quantum circuit framework under topology constraints. We introduced a general approach in which the circuit architecture is determined by the underlying interaction graph, while the variational parameters are optimized to enforce exact uniformity over permutations. Within this framework, we provided explicit controlled-SWAP-based unitary constructions for uniform permutation generation and established rigorous bounds on their circuit complexity. In particular, for linear nearest-neighbor topologies, we obtained a construction with quadratic circuit size and linear depth, thereby removing the need for all-to-all connectivity while improving the depth of previous exact constructions by a linear factor.

Although linear topologies may appear restrictive, the construction extends to any connected architecture that admits an effective linear ordering of subsystems, for example, through a Hamiltonian path or related compact embedding~\cite{bondy2008graph}. Many current quantum platforms satisfy this condition, including superconducting lattices, heavy-hex devices, trapped-ion systems, photonic platforms, and linear ion chains. Thus, the linear circuit 
%\(\mathcal{QC}_l\) 
is implementable on a broad class of existing quantum architectures.

We also identified a fundamental limitation of Bene\v{s}-type architectures. For the quantum Bene\v{s} circuit of size \(n=8\), we concluded, after extensive exploration, that no choice of variational parameters yields the uniform distribution over \(S_8\), despite the logarithmic depth and universal routing capability of the architecture. This leads to conjecture that universal permutation realizability does not automatically imply exact uniform permutation generation.

Several questions remain open. First, it remains to determine whether, for every \(n\), the inductively defined swap matrix \(M_n\) specifies a linear nearest-neighbor quantum circuit that generates a uniform coherent superposition over all permutations in \(S_n\), while using exactly \(\frac{n(n-1)}{2}\) controlled-SWAP operations, circuit depth \(n\), and \(\frac{n(n-1)}{2}\) ancillary control qubits. Second, it is unknown whether the current \(O(n^2)\)-size, \(O(n)\)-depth upper bounds for exact uniform permutation generation are asymptotically optimal, or whether a different circuit architecture can achieve exact uniformity with asymptotically smaller resources. More generally, it remains open whether the true complexity of exact uniform permutation generation lies strictly between the current linear-depth upper bounds and the logarithmic-depth bounds associated with universal permutation realizability.

\appendix

\section{Proof of Lemma \ref{lemma:permutationI}}\label{appenidx1}

\noindent
\textbf{Lemma \ref{lemma:permutationI}}
\emph{Let \(S_n\) denote the symmetric group on \(n\) elements. Define the two sets of disjoint adjacent-swap layers 
\begin{align}
E_1&=\left\{
\prod_{j=1}^{\lfloor \frac{n}{2}\rfloor} W_j
\;\middle|\;W_j\in \{\hat{I},\hat U_{2j-1,2j}\}
\right\},\\ 
E_2 &=\left\{\prod_{j=1}^{\lfloor \frac{n-1}{2} \rfloor} W_j
\;\middle|\; W_j\in \{\hat{I},\hat U_{2j,2j+1}\} \right\}.
\end{align}
Then every permutation \(\pi\in S_n\) can be written as an element of \((E_1\cdot E_2)^k,\) with \(k=\frac n2 \) when \(n\) is even, and as an element of  \((E_1\cdot E_2)^k\cdot E_1\) with \(k= \frac{n-1}{2} \). when \(n\) is odd.
}

\begin{proof}
We first observe that every adjacent transposition \( \hat{U}_{i,i+1} \) is contained in \( E_1 \cdot E_2 \), since all odd- and even-indexed adjacent pairs are included in \( E_1 \) and \( E_2 \), respectively. Therefore, \( E_1 \cdot E_2 \) contains all individual adjacent transpositions.

To construct a general transposition \( \hat{U}_{i,j} \), we use the standard decomposition of a transposition into adjacent swaps:
\begin{equation}\label{eq: Uij decomposition}
    \hat{U}_{i,j} = \hat{U}_{i,i+1} \cdot \hat{U}_{i+1,i+2} \cdots \cdot \hat{U}_{j-1,j} \cdot \hat{U}_{j-2,j-1} \cdots \cdot \hat{U}_{i,i+1},
\end{equation}
which forms a palindromic sequence of adjacent transpositions. This decomposition requires \( 2(j-i)-1 \) adjacent swaps. Since each product \( E_1 \cdot E_2 \) can move a qubit by at most two positions, that is, from \(i\) to \(i\pm 2\), at least \( \left\lceil \frac{j-i}{2} \right\rceil \) layers are required to generate \( \hat{U}_{i,j} \). The worst case corresponds to \( \hat{U}_{1,n} \), which requires \(k \geq \left\lceil \frac{n-1}{2} \right\rceil\). This implies that once the longest transposition is generated, all other transpositions are generated along the way.

Next, we consider disjoint products of transpositions of the form \(\hat{U}_{i_1,j_1}\cdot \hat{U}_{i_2,j_2}\cdots \hat{U}_{i_r,j_r},
\) where all transpositions are pairwise disjoint. Since each transposition \( \hat{U}_{i_s,j_s} \) can be constructed as described above, and since disjointness allows them to be applied in parallel, the total number of layers required is again bounded by the depth of the deepest single transposition. This depth is at most \( \left\lceil \frac{n-1}{2} \right\rceil \). Therefore, all disjoint products of transpositions are contained in \( (E_1 \cdot E_2)^k \).

Finally, we consider \(n\)-cycles, such as \(\hat{U}_{n-1,n} \cdots \hat{U}_{2,3} \cdot \hat{U}_{1,2},\) which can be expressed as a product of adjacent transpositions. Since all transpositions can be constructed using adjacent swaps within \( \left\lceil \frac{n-1}{2} \right\rceil \) layers, this bound also suffices to build any \(n\)-cycle within \( (E_1 \cdot E_2)^k \).

Therefore, all transpositions, disjoint products of transpositions, and \(n\)-cycles—and hence all permutations in \( S_n \)—appear in \( (E_1 \cdot E_2)^k \) for \(k=\left\lceil \frac{n-1}{2} \right\rceil\).
\end{proof}

\bibliographystyle{unsrtnat}
\bibliography{ref.bib}

% \section{gg}
% As discuses in the main part of the paper, the quantum circuit is built by alternating odd and even layer. Then the unitary operator of quantum circuit, for each 

% \begin{equation}
%     \begin{split}
%        \hat{\mathcal{U}}_n=\prod_{\ell=1}^{n/2} \left(\prod_{r=1}^{n/2}\sum_{j_{r}^{(\ell)}=\{0,2r\}}\hat{U}_{j_{r}^{(\ell)},{2r+1}}\sum_{k^{(\ell)}_{r}=\{0,2r-1\}}\hat{U}_{k^{(\ell)}_{r},2r}\right)\otimes
%     \end{split}
% \end{equation}

% \clearpage

% \begin{equation}
% \hat{\mathcal{U}}_n=\prod_{\ell=1}^{n/2}\left(\prod_{r=1}^{n/2-1}\sum_{j_{r}^{(\ell)}=\{0,2r\}}\hat{U}_{j_{r}^{(\ell)},{2r+1}}\prod_{r=1}^{n/2}\sum_{k^{(\ell)}_{r}=\{0,2r-1\}}\hat{U}_{k^{(\ell)}_{r},2r}\right)\otimes\left( \bigotimes_{r=1}^{n/2-1} P_{\epsilon(j_r^{(\ell)})}\otimes\bigotimes_{r=1}^{n/2}P_{\epsilon(k_r^{\ell)})}\right)
% \end{equation}
% with
% \begin{equation}
% \epsilon(x)=
% \begin{cases}
% 0, & x=0,\\
% 1, & x\neq 0,
% \end{cases}
% \qquad
% \hat P_{\epsilon(x)}
% =
% \ket{\epsilon(x)}\bra{\epsilon(x)}.
% \end{equation}
% and \(\hat{U}_{0,r}=I\) for any \(r\).

% Each operator \(\hat{U}_{i,i+1}\) acts on the adjacent subsystems \((i,i+1)\) conditioned on its corresponding control qubit state \(\ket{j_i}\).  When \(j_i = 0\), the operation acts as the identity, and when \(j_i = 1\), it performs a swap between the two neighboring subsystems. This configuration yields a linear circuit depth of \(n\) and a total of \(n(n-1)/2\) controlled-swap gates.

% \clearpage

\end{document}